\documentclass[a4paper]{article}
\usepackage[left=2.3cm,top=1.9cm, bottom=1in, right=2.3cm]{geometry} 
\date{}
\usepackage{amsfonts}
\usepackage{amsmath, amsthm}
\usepackage{graphicx}
\usepackage{hyperref}
\usepackage{float}

\usepackage{graphicx}
\usepackage{amsmath}
\usepackage{amssymb}
\usepackage{amsfonts}
\usepackage{amsopn}
\usepackage{hyperref}

\newtheorem{theorem}{Theorem}

\newtheorem{proposition}[theorem]{Proposition}

\newcommand{\supp}{\text{supp}}
\newcommand{\TD}{\text{\rm TD}}

\begin{document}
\title{On the Optimality of Time Division for Broadcast Channels}

\author{ Salman Beigi\\ 
{\it \small School of Mathematics,} {\it \small Institute for Research in Fundamental Sciences (IPM), Tehran, Iran}\\
{\it \small Department of Information Engineering, The Chinese University of Hong Kong, Hong Kong}}

\date{\today}%

\maketitle

\begin{abstract}
We provide a necessary and sufficient condition that under some technical assumption characterizes all two-receiver broadcast channels for which time division is optimal for transmission of private messages. 
 \end{abstract}


\section{Introduction} 

A two-receiver discrete memoryless broadcast channel $p(y, z| x)$ is a channel with one sender and two receivers. The goal of the sender is to send private messages to the receivers over multiple uses of the channel. The capacity region of the channel is the set of all rate pairs $(R_1, R_2)$ such that private information with asymptotically vanishing error can be sent to the receivers at rate $R_1$ and $R_2$ respectively. 

It is easy to verify that the capacity region of the broadcast channel $p(y, z|x)$ depends only on the \emph{marginal} channels $p(y|x)$ and $p(z|x)$ and not on the whole $p(y, z|x)$. So we may assume with no loss of generality that $p(y, z|x)=p(y|x) p(z|x)$. That is, we may think of a broadcast channel as two point-to-point channels with the same input sets.

There are some known inner and outer bounds for the capacity region of the broadcast channel~\cite{ElGamalKim}, yet deriving a single letter formula for the capacity region is a long standing open problem. 
The best inner bound for the capacity region of the broadcast channel is due to Marton~\cite{Marton} and is described as follows. Let $p_{UVWX}$ be an arbitrary distribution that induces the distribution $p_{UVWXYZ}$. Then any pair of non-negative numbers $(R_1, R_2)$ satisfying 
\begin{align}\label{eq:Marton}
R_1&\leq I(U, W; Y),\nonumber\\
R_2 & \leq I(V, W; Z),\nonumber\\
R_1+R_2&\leq \min\big\{I(W; Y),\, I(W;Z) \big \} + I(U; Y|W) +I(V; Z|W) - I(U; V|W),
\end{align}
is an achievable rate pair. 

The best outer bound on the capacity region of the broadcast channel is called the $UV$ outer bound~\cite{UV}. According to this outer bound for any achievable rate pair $(R_1, R_2)$ there is a distribution $p_{UVX}$ with the induced distribution $p_{UVXYZ}$ such that 
\begin{align}\label{eq:UV}
R_1&\leq I(U; Y),\nonumber\\
R_2 & \leq I(V; Z),\nonumber\\
R_1+R_2&\leq \min\big\{I(U; Y) +I(V; Z|U), \, I(V; Z) +I(U; Y|V)\big\}.
\end{align}

As mentioned above, the Marton inner bound~\eqref{eq:Marton} and the UV outer bound~\eqref{eq:UV} do not match in general, and the capacity region of an arbitrary broadcast channel is not known.

A simple achievable rate region is derived by \emph{time division}.  
Let $C_1=\max_{p_X} I(X; Y)$ be the capacity of the first channel $p(y|x)$ and $C_2=\max_{p_X} I(X; Z)$ be the capacity of the second channel $p(z|x)$. By ignoring the second receiver, the sender can transmit information to the first receiver at the highest possible rate, namely $C_1$. Thus $(R_1, R_2)=(C_1, 0)$ is achievable. Similarly $(R_1, R_2)= (0, C_2)$ is in the capacity region. Moreover, the sender can use time sharing; she can send information to the first receiver in $\alpha\in [0,1]$ fraction of uses of the channel, and then send information to the second receiver in the remaining uses of the channel. Then the rate pair $(R_1, R_2)=(\alpha C_1, (1-\alpha)C_2)$, for any $\alpha\in [0,1]$, is  achievable. More precisely, the whole set 
$$\mathcal R_{\TD}:=\Big\{ (R_1, R_2):\, \frac{R_1}{C_1} + \frac{R_2}{C_2}\leq 1,\, R_1, R_2\geq 0  \Big\},$$
which we call the \emph{time division} rate region, is in the capacity region.

The main result of this paper is a characterization of broadcast channels for which the time division rate region $\mathcal R_{\TD}$ is equal to the capacity region.   Here is an informal statement of our main result. 

\begin{theorem}\label{thm:informal} (Informal) Let $p(y|x)$ and $p(z|x)$ be two point-to-point channels with capacities $C_1$ and $C_2$ respectively. Suppose that $C_1\geq C_2$ and that the channels $p(y|x), p(z|x)$ satisfy some technical assumptions. Then $\mathcal R_{\TD}$ is equal to the capacity region of the broadcast channel $p(y, z|x)=p(y|x) p(z|x)$ if and only if either $C_1< C_2$ and 
\begin{align}\label{eq:main-ineq}
\frac{I(X; Y)}{C_1} \leq  \frac{I(X;Z)}{C_2}, \qquad \forall p_X,
\end{align}
or $C_1=C_2$ and the two channels are more capable comparable. 
\end{theorem}

Recall that a channel $p(y|x)$ is called \emph{more capable} than $p(z|x)$ if for all input distributions $p_X$ we have
$$I(X; Y) \geq  I(X;Z).$$ 
We say that two channels $p(y|x)$ and $p(z|x)$ are more capable comparable if either $p(y|x)$ is more capable than $p(z|x)$ or vice versa. 

A partial characterization of \emph{degraded} broadcast channels for which time division is optimal is provided in~\cite[Theorem 3]{GGA14} that is similar to our characterization.

To prove this theorem we use some known facts about the set of \emph{capacity achieving distributions} of a point-to-point channel~\cite[Theorem 13.1.1]{CoverThomas}. For the convenience of the reader we also present the proofs of these facts.

\section{Capacity achieving distributions}

We use quite standard notations in this paper (see, e.g.,~\cite{ElGamalKim}). Sets are denoted by calligraphic letters such as $\mathcal X$, and a distribution on such a set is specified by a subscript as in $p_X$. Discrete memoryless point-to-point channels are determined by a set of conditional distributions $\{p_{Y|x}, x\in \mathcal X\}$, on a set $\mathcal Y$. For ease of notation we denoted such a channel by $p(y|x)$. For $\lambda\in [0,1]$ we use the notation $\bar \lambda:=1-\lambda$. To avoid confusions, when a mutual information $I(X; Y)$ is computed with respect to a distribution $p_{XY}$ we denote it by $I(X; Y)_p$.

Let $p(y|x)$ be a discrete memoryless point-to-point channel with capacity $C=\max_{p_X} I(X; Y)$. 
For arbitrary distributions $p_X$ and $r_Y$ define 
$$\psi(p_X, r_Y) = \sum_x p(x)D(p_{Y|x} \| r_Y ) = D(p_{XY}\| p_X r_Y) = I(X; Y)_p + D(p_Y\| r_Y),$$
where $D(\cdot \| \cdot)$ is the KL divergence, and $p_{XY}$ is the induced distribution on the input and output of the channel with input distribution $p_X$. 
Then by the joint convexity of KL divergence and Sion's minimax theorem we have 
$$\max_{p_X} \min_{r_Y} \psi(p_X, r_Y) =  \min_{r_Y} \max_{p_X} \psi(p_X, r_Y).$$
Let us compute each side of the above equation. By the non-negativity of KL divergence we have
$$\max_{p_X} \min_{r_Y} \psi(p_X, r_Y)  =\max_{p_X} \min_{r_Y}  I(X; Y)_p + D(p_Y\| r_Y) =  \max_{p_X} I(X; Y) = C.$$
On the other hand, by the linearity of $\psi$ in $p_X$ have 
$$\min_{r_Y} \max_{p_X} \psi(p_X, r_Y) = \min_{r_Y} \max_{x_0} D(p_{Y|x_0}\| r_Y).$$
As a result, 
\begin{align}\label{eq:min-max-q-Y} 
\min_{r_Y} \max_{x_0} D(p_{Y|x_0}\| r_Y) = C.
\end{align}

Observe that the minimum in~\eqref{eq:min-max-q-Y} is achieved. So
let $r^*_Y$ be some optimal distribution there, i.e., $r^*_Y$ is such that 
$$\max_{x_0} D(p_{Y|x_0}\| r^*_Y)=\max_{p_X} I(X; Y)_p + D(p_Y\| r^*_Y)=C.$$
Then for every $p_X$ we have
\begin{align}\label{eq:CA-ineq} 
\psi(p_X, r^*_Y) =I(X; Y)_p + D(p_Y\| r^*_Y)\leq C.
\end{align}

Let $\Pi$ be the set of capacity achieving distributions:
$$\Pi=\arg\max_{p_X} I(X; Y).$$
Then by~\eqref{eq:CA-ineq} for every $p_X\in\Pi$ we have $D(p_Y\| r^*_Y) = 0$, i.e., $p_Y=r^*_Y$.  This means that, for any capacity achieving distribution $p_X\in \Pi$ its induced distribution on the output of the channel is fixed, i.e., $r^*_Y=p_Y$. Indeed, the \emph{optimal output distribution} of a channel is unique, and is the unique distribution $r_Y^*$ that achieves the minimum in~\eqref{eq:min-max-q-Y}. 

Let us define
\begin{align}\label{eq:def-K}
\mathcal K:=\{x:\, D(p_{Y|x}\| r^*_Y) = C  \}.
\end{align}
Note that by the above discussion $\mathcal K$ is non-empty. Let $p_X$ be some distribution with $\supp(p_X)\subseteq \mathcal K$, where 
$\supp(p_X):=\{x:\, p(x)>0 \}$. Then we have
\begin{align*}
I(X; Y)_p + D(p_Y\| r^*_Y) = \sum_x p(x) D(p_{Y|x}\| r^*_Y) =  \sum_{x\in \mathcal K} p(x) D(p_{Y|x}\| r^*_Y) = C.
\end{align*}
This means that the inequality in~\eqref{eq:CA-ineq} becomes an equality for all $p_X$ with $\supp(p_X)\subseteq \mathcal K$.

On the other hand, let $p_X\in \Pi$ be some capacity achieving distribution. Then by the above discussion, $p_Y = r^*_Y$. Moreover, we have
\begin{align*}
C & = I(X; Y)_p = \sum_x p(x) D(p_{Y|x}\| r^*_Y) \\
&= \sum_{x\in \mathcal K} p(x) D(p_{Y|x}\| r^*_Y) + \sum_{x\notin \mathcal K} p(x) D(p_{Y|x}\| r^*_Y) \\
& = p(\mathcal K) C + p(\mathcal X\setminus \mathcal K) \max_{x\notin \mathcal K} D(p_{Y|x}\| r^*_Y).
\end{align*}
As a result, we must have $p(\mathcal X\setminus \mathcal K) = 0$, i.e., $\supp(p_X) \subseteq \mathcal K$. 

Finally, suppose that $p_X$ is some distribution with $\supp(p_X)\subseteq \mathcal K$ and $p_Y=r^*_Y$. Then~\eqref{eq:CA-ineq} becomes equality for $p_X$ and since $D(p_Y\| r^*_Y)=0$, $p_X$ is capacity achieving. 

We summarize the above findings in the following proposition.

\begin{proposition}\label{prop:def-Pi}
For any point-to-point channel $p(y|x)$ there is a unique distribution $r^*_Y$ such that for all capacity achieving distributions $p_X\in \Pi$ we have $p_Y=r^*_Y$. Moreover, for any $p_X\in \Pi$ we have $\supp(p_X)\subseteq \mathcal K$ where $\mathcal K$ is defined in~\eqref{eq:def-K}. Indeed, a given distribution $p_X$ is capacity achieving if and only if $\supp(p_X)\subseteq \mathcal K$ and $p_Y= r^*_Y$. In particular $\Pi$ is convex. 
\end{proposition}

The above proposition motivates the following definition. Define $\mathcal K_0$ to be the union of the supports of capacity achieving distributions, i.e.,
\begin{align}\label{eq:def-K-0}
\mathcal K':=\bigcap_{p_X\in \Pi} \supp(p_X).
\end{align}
By the above proposition $\mathcal K'\subseteq \mathcal K$.

For a channel that has a capacity achieving distribution with full support (e.g., a channel for which the uniform distribution is capacity achieving) we have $\mathcal K'=\mathcal K=\mathcal X$. For example, this equality holds for binary symmetric and binary erasure channels. 
Later we will see an example of a channel for which the inclusion $\mathcal K'\subseteq \mathcal K$ is strict.

\begin{proposition}\label{prop:K-p}
There exists $r_X\in \Pi$ such that $\supp(r_X)=\mathcal K'$. Moreover, for any $p_X$ with $\supp(p_X)\subseteq \mathcal K'$ we have
$$I(X; Y)_p + D(p_Y\| r^*_Y)=\sum_{x\in \mathcal K'} p(x)D(p_{Y|x} \| r^*_Y )  = C.$$
\end{proposition}

\begin{proof}
The existence of $r_X\in \Pi$ with $\supp(r_X)=\mathcal K'$ follows from the definition of $\mathcal K'$ and the convexity of $\Pi$ established in Proposition~\ref{prop:def-Pi}.  The second claim follows from $\mathcal K'\subseteq \mathcal K$.
\end{proof}

\section{Proof of the main result}
Let $p(y|x)$ and $p(z|x)$ be two channels with capacities $C_1$ and $C_2$ respectively. Let $r^*_Y$ and $s^*_Z$ be the optimal output distributions of the channels (as defined in the previous section) respectively. Also let $\Pi_1$ and $\Pi_2$ be  their associated sets of capacity achieving distributions respectively. Finally let $\mathcal K_1, \mathcal K'_1$ and $\mathcal K_2, \mathcal K'_2$ be their associated subsets of $\mathcal X$ defined by~\eqref{eq:def-K} and~\eqref{eq:def-K-0}. 
Here is the formal statement of our main result.

\begin{theorem}\label{thm:main-thm} 
Suppose that $\mathcal K'_1=\mathcal K'_2=\mathcal X$ and $C_1\geq C_2$. Then the time division rate region $\mathcal R_\TD$ is the capacity region of the broadcast channel $p(y, z|x)=p(y|x)p(z|x)$ if and only if either $C_1< C_2$ and 
\begin{align}\label{eq:main-ineq}
\frac{I(X; Y)}{C_1} \leq  \frac{I(X;Z)}{C_2}, \qquad \forall p_X,
\end{align}
or $C_1=C_2$ and the two channels are more capable comparable. 
\end{theorem}

\vspace{.09in}

The rest of this section is devoted to the proof of this theorem.\\

\noindent 
$(\Rightarrow)$ First suppose that the time division region is the capacity region. That is, for any achievable rate pair $(R_1, R_2)$ we have 
\begin{align}\label{eq:TD}
\frac{R_1}{C_1} + \frac{R_2}{C_2}\leq 1.
\end{align}

Let $r_{XU}$ and  $s_{XV}$ be arbitrary distributions.
Define $p_{Q W \widetilde U \widetilde V X}$ by 
\begin{align}\label{eq:margin-Q} 
p(Q=0)=\lambda, \qquad p(Q=1)=\bar \lambda= 1-\lambda,
\end{align}
and according to the following table:
\begin{center}
\begin{tabular}{c|c|c|c}
         & $W$ & $\widetilde U$   & $\widetilde V$ \\
         \hline
 $Q=0$ & $U$ &   $X$ & \text{Const.}\\
 $Q=1$ & $V$ &   \text{Const.} & $X$
\end{tabular}
\end{center}
\label{default}
This table should be understood as follows. Firstو we have $\mathcal Q=\{0,1\}$ and the marginal distribution $p_Q$ is given by~\eqref{eq:margin-Q}. Second, we have $\mathcal W = \mathcal U\, \dot{\cup}\,\mathcal V$, $\widetilde{\mathcal U} = \mathcal X\,\dot{\cup}\,\{u^*\}$ and $\widetilde{\mathcal V} = \mathcal X\,\dot{\cup}\, \{v^*\}$ for two distinguished elements $u^*, v^*$. Third, the conditional distribution $p(w, \tilde u, \tilde v, x|q)$ is given by
\begin{align*}
p(w, \tilde u, \tilde v, x| Q=0) = \begin{cases}
 r(X=x, U=w), & \tilde u=x, \tilde v=v^*, w\in \mathcal U,   
\\
0, & \text{ otherwise,} 
\end{cases}
\end{align*}
and
\begin{align*}
p(w, \tilde u, \tilde v, x| Q=1) = \begin{cases}
 s(X=x, V=w), & \tilde v=x, \tilde u=u^*, w\in \mathcal V,  
\\
0, & \text{ otherwise.} 
\end{cases}
\end{align*}

Now  let $\widetilde W = (Q, W)$ and consider the distribution $p_{\widetilde W\widetilde U\widetilde V XYZ}$ induced by the channel. Observe that $I(\widetilde U; \widetilde V| \widetilde W) = I(\widetilde U; \widetilde V| Q, W) =0$. Then by Marton's coding theorem~\eqref{eq:Marton} the rate pair $(R_1, R_2)$ given by 
\begin{align*}
\begin{cases}
R_2= I(\widetilde V \widetilde W; Z) = I(\widetilde W; Z) + I(\widetilde V; Z|\widetilde W),
\\ R_1+ R_2 = \min\big\{ I(\widetilde W; Y), I(\widetilde W; Z)    \big\} + I(\widetilde U; Y| \widetilde W) + I(\widetilde V; Z| \widetilde W), 
\end{cases}
\end{align*}
is achievable. 
Therefore, by our assumption we must have 
\begin{align*}
\frac{R_1}{C_1} + \frac{R_2}{C_2} = \frac{R_1+R_2}{C_1} + \big(\frac{1}{C_2} - \frac{1}{C_1} \big) R_2 \leq 1,
\end{align*}
and then
\begin{align}\label{eq:bound-21}
 \frac{1}{C_1} \min\big\{ I(\widetilde W; Y), I(\widetilde W; Z) \big\} +  \big(\frac{1}{C_2} - \frac{1}{C_1} \big) I(\widetilde W; Z)  + \frac{1}{C_1}I(\widetilde U; Y| \widetilde W)  + \frac{1}{C_2}I(\widetilde V; Z|\widetilde W)\leq 1.
\end{align}
Let us compute individual terms in the above equation. We have
\begin{align}\label{eq:I-w-y}
I(\widetilde W; Y) = I(Q, W; Y) = I(Q; Y) + I(W; Y|Q) = I(Q; Y) + \lambda I(U; Y)_r + \bar \lambda I(V; Y)_s.
\end{align}
We similarly have
\begin{align}\label{eq:I-w-z}
I(\widetilde W; Z) = I(Q; Z) + \lambda I(U; Z)_r + \bar\lambda I(V; Z)_s.
\end{align}
Moreover, observe that
\begin{align*}
I(\widetilde U; Y| \widetilde W) & = \lambda I(X; Y| U)_r,
\end{align*}
and
\begin{align*}
I(\widetilde V; Z| \widetilde W) & = \bar \lambda I(X; Z| V)_s.
\end{align*}
Putting the above two equations in~\eqref{eq:bound-21} we find that 
\begin{align}\label{eq:gen-ineq}
 \frac{1}{C_1} \min\big\{ I(\widetilde W; Y), I(\widetilde W; Z) \big\} +  \big(\frac{1}{C_2} - \frac{1}{C_1} \big) I(\widetilde W; Z)  + \frac{\lambda}{C_1}I(X; Y| U)_r  + \frac{\bar \lambda}{C_2}I(X; Z|V)_s\leq 1.
\end{align}
We can now consider two cases: either $I(\widetilde W; Y) \geq I(\widetilde W; Z)$ or $I(\widetilde W; Y) < I(\widetilde W; Z)$. 
Then using $C_1\geq C_2$, equations~\eqref{eq:I-w-y} and~\eqref{eq:I-w-z}, and ignoring some non-negative terms in~\eqref{eq:gen-ineq} we find that 
\begin{align*}
\begin{cases}
\frac{1}{C_1} \Big(\lambda I(U; Y)_r +\bar \lambda I(V; Y)_s \Big) +\big(\frac{1}{C_2} - \frac{1}{C_1} \big)\bar\lambda I(V; Z)_s+ \frac{\lambda}{C_1}I(X; Y| U)_r  + \frac{\bar \lambda}{C_2}I(X; Z|V)_s\leq 1,\\
~\text{or}\\
\frac{1}{C_2}   \Big(\lambda I(U; Z)_r +\bar \lambda I(V; Z)_s \Big)  + \frac{\lambda}{C_1}I(X; Y| U)_r  + \frac{\bar \lambda}{C_2}I(X; Z|V)_s\leq 1.
\end{cases}
\end{align*}

Now suppose that we chose  $r_{XU}$ and $s_{XV}$ such that $r_X\in \Pi_1$ and $s_X\in \Pi_2$. Then using $C_1=I(X; Y)_r = I(UX; Y)_r = I(U; Y)_r + I(X; Y|U)_r$ and $C_2=I(X; Z)_s = I(VX; Z)_s = I(U; Z)_s + I(X; Y|V)_s$, by a simple algebra we arrive at 
\begin{align*}
\begin{cases}
 I(V; Y)_s \leq  I(V; Z)_s,\\
~\text{or}\\
\frac{1}{C_2}  I(U; Z)_r \leq  \frac{1}{C_1}I(U; Y)_r.
\end{cases}
\end{align*}
Observe that the first inequality here depends only on $s_{XV}$ and the second one is solely in terms of $r_{XU}$. Then either the first one holds for every valid choice of $s_{XV}$ or the second one holds for every valid choice of $r_{XU}$. This means that either 
\begin{align}\label{eq:O1}
I(V; Y) \leq  I(V; Z), \qquad \forall s_{VX} \,\text{ s.t. }\, s_X\in \Pi_2,
\end{align}
or
\begin{align}\label{eq:O2}
\frac{1}{C_2}I(U; Z) \leq  \frac{1}{C_1}I(U; Y), \qquad \forall r_{UX} \,\text{ s.t. }\, r_X\in \Pi_1.
\end{align}

Let us suppose that~\eqref{eq:O1} holds. Fix $s_X\in \Pi_2$ to be a capacity achieving distribution for $p(z|x)$ with $\supp(s_X)=\mathcal K'_2=\mathcal X$ whose existence is guaranteed by Proposition~\ref{prop:K-p}.  Let $p_X$ be an arbitrary distribution. Define $s_{VX}$ as follows. Let $\mathcal V=\{0,1\}$ and define $s(V=0)=\epsilon$ and $s(V=1)=1-\epsilon$. Also let 
$$s(x|V=0)=p(x), \qquad s(x|V=1)=\frac{1}{1-\epsilon} s(x) - \frac{\epsilon}{1-\epsilon}p(x).$$
Observe that $\supp(p_X)\subseteq \supp(s_X)=\mathcal X$, so for sufficiently small $\epsilon>0$, both $s(x|V=0)$ and $s(x|V=1)$ are valid distributions. Then we obtain a distribution $s_{VX}$ on $\{0,1\}\times \mathcal X$ whose marginal on $\mathcal X$ is the distribution $s_X\in \Pi_2$ we started with.  Since we assumed that~\eqref{eq:O1} holds, for any sufficiently small $\epsilon>0$ we have $I(V; Y)\leq I(V;Z)$. Now a simple computation verifies that $I(V;Y) = \epsilon D(p_Y\| s_Y)+\Theta(\epsilon^2)$ and $I(V; Z) = \epsilon D(p_Z\| s_Z) + \Theta(\epsilon^2)$.  Therefore, we have
\begin{align}\label{eq:op-01}
D(p_Y\| s_Y) \leq  D(p_Z\| s_Z^*), \qquad \forall p_X.
\end{align}

Starting from~\eqref{eq:O2} and following similar arguments we find that  
\begin{align}\label{eq:op-02}
\frac{1}{C_2}  D(p_Z \| r_Z) \leq  \frac{1}{C_1}D(p_Y\| r^*_Y), \qquad \forall p_X, 
\end{align}
where $r_X\in \Pi_2$ is a capacity achieving distribution for $p(y|x)$ with $\supp(r_X)=\mathcal K'_1=\mathcal X$. Then either~\eqref{eq:op-01} or~\eqref{eq:op-02} is satisfied.

Let us in~\eqref{eq:op-01} and~\eqref{eq:op-02} restrict ourself to $p_X$ of the form $p(x) = \delta_{x, x_0}$, where $\delta_{x, x_0}$ denotes the Kronecker delta function and $x_0\in \mathcal X$ is arbitrary. Then either
\begin{align}\label{eq:op-11}
D(p_{Y|x_0}\| s_Y) \leq  D(p_{Z|x_0}\| s_Z^*)=C_2, \qquad \forall x_0 ,
\end{align}
or
\begin{align}\label{eq:op-12}
\frac{1}{C_2}  D(p_{Z|x_0} \| r_Z) \leq  \frac{1}{C_1}D(p_{Y|x_0}\| r^*_Y) = 1, \qquad \forall x_0,
\end{align}
holds.

If~\eqref{eq:op-01} and then~\eqref{eq:op-11} hold, we have $\max_{x_0} D(p_{Z|x_0} \| s_Y) \leq C_2$, which using~\eqref{eq:min-max-q-Y}  gives $C_1\leq C_2$. Then by our assumption $C_1\geq C_2$ we arrive at $C_1=C_2$. Therefore,~\eqref{eq:op-01} does not hold if $C_1> C_2$.  Moreover, if $C_1=C_2$,~\eqref{eq:op-01} and~\eqref{eq:op-02} are symmetric. So in both cases, with no loss of generality we may assume that~\eqref{eq:op-02} and then~\eqref{eq:op-12} are satisfied. 

We note that~\eqref{eq:op-12} implies that $\max_{x_0} D(p_{Z|x_0} \| r_Z)\leq C_2$. Thus $r_Z$
 is an optimal output distribution for $p(z|x)$, and by its uniqueness $r_Z=s^*_Z$. Then~\eqref{eq:op-02} reduces to
$$\frac{1}{C_2}  D(p_Z \| s^*_Z) \leq  \frac{1}{C_1}D(p_Y\| r^*_Y),\qquad \forall p_X.$$
On the other hand, by Proposition~\ref{prop:K-p} we have $D(p_Y\| r^*_Y) = C_1- I(X; Y)_p$ and $D(p_Z\| s^*_Z)= C_2-I(X; Z)_p$.
Using these in the above inequality gives~\eqref{eq:main-ineq}.

\vspace{.2in}

\noindent $(\Leftarrow)$ We now prove the converse. We assume that either $C_1> C_2$ and~\eqref{eq:main-ineq} holds, or $C_1=C_2$ and the two channels are more-capable comparable. In the latter case, by symmetry with no loss of generality we assume that $p(y|x)$ is more capable than $p(z|x)$. Then in both cases~\eqref{eq:main-ineq} holds. Assuming this
we show that time division is optimal for the broadcast channel $p(y|x)p(z|x)$.  

Let $(R_1, R_2)$ be an achievable rate pair. By the $UV$ outer bound~\eqref{eq:UV}, there exists $p_{UVX}$ such that 
\begin{align*}
R_2&\leq I(V; Z),\\
R_1+R_2&\leq I(V; Z) + I(U; Y|V)\leq I(V; Z) + I(X; Y|V).
\end{align*}
Then using the fact that $C_1\geq C_2$ we find that 
\begin{align*}
\frac{R_1}{C_1} + \frac{R_2}{C_2} &\leq \frac{I(X; Y|V)}{C_1} + \frac{I(V; Z)}{C_2}\\
& \leq \frac{I(X; Z|V)}{C_2} + \frac{I(V; Z)}{C_2}\\
& \leq \frac{I(X; Z)}{C_2}\\
& \leq 1,
\end{align*}
where in the second line we use~\eqref{eq:main-ineq}. We are done.

\section{Example}

In the statement of Theorem~\ref{thm:main-thm} we assume that $\mathcal K'_1=\mathcal K'_2 = \mathcal X$. This may seem an unnecessary technical assumption that is forced by our proof method.  
Here we give an example to illustrate that Theorem~\ref{thm:main-thm} does not hold without it.

Let $\mathcal A, \mathcal B$ be two finite disjoint sets with $|\mathcal A| \geq |\mathcal B|\geq 2$. Let $\mathcal X=\mathcal A\cup \mathcal B$, $\mathcal Y=\mathcal A$ and $\mathcal Z=\mathcal B$. Define channels $p(y|x)$ and $p(z|x)$ by
\begin{align*}
p(y| x) = 
\begin{cases}
\delta_{x,y}  &\quad x\in \mathcal A,\\
\frac{1}{|\mathcal B|} &\quad   x\in \mathcal B,
\end{cases}
\qquad\qquad 
p(z| x) = 
\begin{cases}
\frac{1}{|\mathcal B|}  &\quad x\in \mathcal A,\\
\delta_{z, x} &\quad   x\in \mathcal B,
\end{cases}
\end{align*}
where $\delta$ denotes the Kronecker delta function.  

The uniform distribution on the subset $\mathcal A\subset \mathcal X $ is the unique capacity achieving distribution of $p(y|x)$, and its optimal output distribution is the uniform distribution on $\mathcal Y$. Then $\mathcal K'_1 = \mathcal A$. Likewise the unique capacity achieving distribution of $p(z|x)$ is the uniform distribution on $\mathcal B\subset \mathcal X$ and we have $\mathcal K'_2 = \mathcal B$.
Moreover, $C_1=\log |\mathcal A| \geq \log |\mathcal B|= C_2$. 
Also note that~\eqref{eq:main-ineq} does not hold, nor the two channels are more-capable comparable. Indeed, if $p_X$ is a non-trivial distribution supported only on $\mathcal A$, then $I(X; Y)_p>0$ while $I(X; Z)_p=0$. Similarly if $p_X$ is non-trivial and supported only on $\mathcal B$, then $I(X; Y)_p=0$ while $I(X; Z)_p>0$.
Nevertheless, we show in the following that time division is optimal for the broadcast channel $p(y|x)p(z|x)$.

Let $(R_1, R_2)$ be an achievable rate pair. Again by the $UV$ outer bound~\eqref{eq:UV}, there exists $p_{UVX}$ such that 
\begin{align*}
R_2&\leq I(V; Z),\\
R_1+R_2&\leq I(V; Z) + I(U; Y|V)\leq I(V; Z) + I(X; Y|V).
\end{align*}
Then using the fact that $C_1\geq C_2$ we find that 
\begin{align}\label{eq:A123}
\frac{R_1}{C_1} + \frac{R_2}{C_2} &\leq \frac{I(X; Y|V)}{C_1} + \frac{I(V; Z)}{C_2}\nonumber\\
& \leq \frac{I(X; Y)}{C_1} + \frac{I(V; Z)}{C_2},
\end{align}
where in the second line we use the fact that $V-X-Y$ forms a Markov chain.

Let $Q$ be a binary random variable that equals $0$ if $X\in \mathcal A$ and equals $1$ if $X\in \mathcal B$. 
Then we have
\begin{align*}
I(X; Y) & = I(XQ; Y)\\
&  = H(Y) - p(Q=0) H(Y| X, Q=0) -p(Q=1) H(Y| X, Q=1)\\
& = H(Y) - p(Q=1) \log |\mathcal A |\\
& \leq p(Q=0) C_1.
\end{align*}
We similarly have $I(X; Z) \leq p(Q=1) C_2$. Putting these in~\eqref{eq:A123} we find that $R_1/C_1 + R_2/C_2\leq 1$. Therefore, time division is optimal for $p(y|x)p(z|x)$.


\vspace{.23in}
\noindent\textbf{Acknowledgements.}   The author is thankful to Chandra Nair for several fruitful discussions about broadcast channels. The author was supported in part by Institute of Network Coding of CUHK and by GRF grants 2150829 and 2150785.


\end{document}